\documentclass[a4paper,USenglish]{article}%{lipics-v2021}
\usepackage[utf8]{inputenc}
\providecommand{\email}[1]{\href{mailto:#1}{\nolinkurl{#1}\xspace}}
\usepackage{amsmath,amssymb,amsfonts,amsthm}
\usepackage{graphicx}
\usepackage{fullpage}
\usepackage[backref=page]{hyperref}
\usepackage[capitalize]{cleveref}
\usepackage{color}
\usepackage{wrapfig}
\usepackage{tikz}
\usetikzlibrary{shapes, arrows, tikzmark, calc}
\usepackage{setspace}
\usepackage{algorithm}
\usepackage[noend]{algpseudocode}
\usepackage[framemethod=tikz]{mdframed}
\usepackage{xspace}
\usepackage{pgfplots}
\usepackage{framed}
\usepackage{subcaption}
 \pgfplotsset{compat=1.5}
  \newcommand{\eps}[0]{\ensuremath{\varepsilon}}
  \let\epsilon\eps
  
  \newcommand{\cA}{\ensuremath{{\mathcal A}}\xspace}
  \newcommand{\cB}{\ensuremath{{\mathcal B}}\xspace}
  \newcommand{\cC}{\ensuremath{{\mathcal C}}\xspace}

  \newcommand{\cP}{\ensuremath{{\mathcal P}}\xspace}
  \newcommand{\cQ}{\ensuremath{{\mathcal Q}}\xspace}
  \newcommand{\cR}{\ensuremath{{\mathcal R}}\xspace}
  \newcommand{\cS}{\ensuremath{{\mathcal S}}\xspace}

  \newcommand{\cX}{\ensuremath{{\mathcal X}}\xspace}

  \newcommand{\bbF}{\ensuremath{{\mathbb F}}\xspace}

  \newcommand{\bbN}{\ensuremath{{\mathbb N}}\xspace}

  \newcommand{\bbR}{\ensuremath{{\mathbb R}}\xspace}

\makeatletter
\newtheorem*{rep@theorem}{\rep@title}
\newcommand{\newreptheorem}[2]{%
\newenvironment{rep#1}[1]{%
 \def\rep@title{\theoremref{##1} Restated}%
 \begin{rep@theorem}}%
 {\end{rep@theorem}}}
\makeatother

\makeatletter
\newtheorem*{rep@lemma}{\rep@title}
\newcommand{\newreplemma}[2] {%
\newenvironment{rep#1}[1]{%
 \def\rep@title{\lemmaref{##1} Restated}%
 \begin{rep@lemma}}%
 {\end{rep@lemma}}}
\makeatother

\newtheorem{theorem}{Theorem}
\newreptheorem{theorem}{Theorem}

\newtheorem{definition}{Definition}
\newtheorem{lemma}{Lemma}

\newreplemma{lemma}{Lemma}

\newtheorem{corollary}[theorem]{Corollary}

\theoremstyle{definition}
\newtheorem*{remark}{Remark}
\newcommand{\namedref}[2]{\hyperref[#2]{#1~\ref*{#2}}\xspace}

\newcommand{\lemmaref}[1]{\namedref{Lemma}{lem:#1}}

\newcommand{\theoremref}[1]{\namedref{Theorem}{thm:#1}}

\title{Differential Privacy and Sublinear Time are Incompatible Sometimes}

\author{Jeremiah Blocki\\Purdue University\\ jblocki@purdue.edu
\and 
Hendrik Fichtenberger\\ 
Google Research\\fichtenberger@google.com
\and 
Elena Grigorescu\\Purdue University \\elena-g@purdue.edu
\and Tamalika Mukherjee\\Columbia University\\ tm3391@columbia.edu
}%{https://orcid.org/0000-0003-0052-4181}{}

\begin{document}

\maketitle
\begin{abstract}
    Differential privacy and sublinear algorithms are both rapidly emerging algorithmic themes in times of big data analysis. Although recent works have shown the existence of differentially private sublinear algorithms for many problems including graph parameter estimation and clustering, little is known regarding hardness results on these algorithms. In this paper, we initiate the study of lower bounds for problems that aim for both differentially-private and sublinear-time algorithms. Our main result is the incompatibility of both the desiderata in the general case. In particular, we prove that a simple problem based on one-way marginals yields both a differentially-private algorithm, as well as a sublinear-time algorithm, but does not admit a  ``strictly'' sublinear-time algorithm that is also differentially private. 
\end{abstract}

\section{Introduction}

While individuals have long demanded privacy-preserving analysis and processing of their data, their adoption and enforcement by governmental and private standards, policies, and jurisdictions are now accelerating. This urgency stems, in part, from the dramatic growth in the amount of data collected, aggregated, and analyzed per individual in recent years. The sheer volume of data also poses a computational challenge as resource demands scale with data size. Thus, it is expedient to develop privacy preserving algorithms for data-analysis whose resource requirements scale sub-linearly in the size of the input dataset. Two algorithmic concepts that formalize these two objectives are \emph{differential privacy} (DP) and \emph{sublinear algorithms}. A randomized algorithm is differentially private if its output distribution does not change significantly when we slightly modify the input dataset to add/remove an individual's data, i.e., a row in the dataset. Sublinear algorithms comprise classes of algorithms that have time or space complexity that is sublinear in their input size. Previous work in the intersection of both fields has promoted several classical sublinear algorithms to differentially private sublinear time algorithms for the same problems, e.g., sublinear-time clustering~\cite{blocki2021differentially}, graph parameter estimation~\cite{blocki2022privately} or sublinear-space heavy hitters in streaming~\cite{pmlr-v97-upadhyay19a,BlockiLMZ23}. 

Intuitively, one might expect that privacy and sublinear time necessarily have a symbiotic relationship, i.e., if only a fraction of the data is processed, a significant amount of sensitive information may remain unread. Recent work~\cite{BlockiGMZ23} demonstrated that if a function $f \colon D \to \bbR$ has low global sensitivity (i.e., $f$ is amenable to DP) and there exists a {\em sufficiently accurate} sublinear-time approximation algorithm for $f$, then there exists an accurate sublinear time DP approximation for $f$. This lead \cite{BlockiGMZ23} to ask whether or not a similar transformation might apply for functions $f \colon D \to \bbR^d$ with multi-dimensional output. In this paper, we provide an example of a function $f \colon D \to \bbR^d$ with the following properties (1) there is an efficient sublinear time approximation algorithm for $f$, (2) there is a differentially private approximation algorithm for $f$ running in time $O(|D|)$, and (3) any accurate differentially private approximation algorithm must run in time $\tilde{\Omega}(|D|)$. Thus, the intuition that privacy and sublinear time algorithms necessarily have a symbiotic relationship is incorrect.

{\bfseries Existing Model for DP Lower Bounds. }Consider a database $D \in \{0,1\}^{n \times d}$ with $n$ rows where each row corresponds to an individual's record with $d$ binary attributes. The model in this case is that the database $D$ consists of a sample (according to a uniform or possibly adversarial distribution) from a larger population (or universe), and we are interested in answering queries on the sample (e.g., what fraction of individual records in $D$ satisfy some property $q$?). One of the interesting questions in DP lower bounds is the following: Suppose we fix a family of queries $\cQ$ and the dimensionality of database records $d$, then what is the sample complexity required to achieve DP and statistical accuracy for $\cQ$? Here the \emph{sample complexity} is defined as the minimum number of records $n$ such that there exists a (possibly computationally unbounded) algorithm that achieves both DP and accuracy. 

A key problem that has been at the center of addressing this question is the \emph{one-way marginal problem}. The one-way marginal problem takes a database $D \in \{0,1\}^{n \times d}$ and releases the average values of all $d$ columns. The best known private algorithm which has running time polynomial in the universe size is based on the multiplicative weights mechanism, and it achieves $(O(1), o(1/n))$-DP for $n \in O(\sqrt{d} \log \vert \cQ\vert)$~\cite{HardtR10}. Any pure DP algorithm for this problem requires $n \geq \Omega(d)$ samples~\cite{HardtT10}, while \cite{BunUV18} showed that $n \geq \tilde{\Omega}(\sqrt{d})$ samples are necessary to solve this problem with approximate DP and within an additive error of 1/3. 

{\bfseries Existing Models for Sublinear-Time Algorithms. } The works on sublinear-time algorithms utilize different input models, many of them tailored to the representation of the input, e.g., whether it is a function or a graph. These models typically define query oracles, i.e., mechanisms to access the input in a structured way. For example, the dense graph model~\cite{goldreich1998property} defines an adjacency query oracle that, given a pair of indices $(i,j)$, returns the entry $A(i,j)$ of the adjacency matrix $A$ of the input graph. Query oracles enable an analysis of what parts, and more generally how much of the input was accessed by an algorithm. Since the fraction of input read is a lower bound for the time complexity of an algorithm, query access models are crucial to prove both sublinear-time upper and lower bounds.

{\bfseries Our Model. }A challenge of proving lower bounds for sublinear time, differentially private algorithms lies in devising and applying a technique for analysis that combines the properties of both models. Lower bounds for differential privacy state a lower bound on the sample complexity that is required to guarantee privacy and non-trivial accuracy. These bounds do not state how much of the input an algorithm needs to \emph{read} to guarantee privacy and accuracy, but only what input size \emph{is required} to (potentially) enable such an algorithm. On the other hand, lower bounds for sublinear time algorithms state a bound on the time complexity as a function of the input size $m$.
Note that time complexity is at least query complexity, and a lower bound on the latter immediately implies a lower bound on time complexity as well.

In our setting, we fix the number of records $n$, as well as the dimensionality $d$ of the database, i.e., our problem size is $m=n\cdot d$. We define queries in our model to be \emph{attribute queries}, i.e., querying the $j$-th attribute of a row $i$ in the database $D$ is denoted as $D(i,j)$. We emphasize that our use of the term \emph{query} in our model is as in the sublinear-time algorithms model, and it is different from its use in conventional DP literature. Specifically, queries in DP literature refer to types of questions that the data analyst can make to the database to infer something about the population, whereas queries in the sublinear algorithms model refer to how the algorithm can access the input dataset $D$. In our work, we fix a problem of interest $\cP$ on database $D$ (e.g., the one-way marginal problem), and we consider an algorithm that solves the problem $\cP$ on input $D$. Then we are interested in understanding the minimum number of (attribute) queries that an algorithm can make to solve the problem $\cP$ and satisfy both DP and accuracy, which we call the \emph{query complexity}.

{\bfseries Result of \cite{BunUV18} does not apply in our model. }For the problem of one-way marginals, we know that $n \in \tilde{\Omega}(\sqrt{d})$~\cite{BunUV18} records are required for any algorithm to achieve both DP and accuracy. For $m=\tilde{\Omega}(d^{3/2})$, there exists a DP algorithm that can solve this problem with $\tilde{O}(m)$ queries, i.e., the algorithm can query the entire dataset and add Gaussian noise. Using Hoeffding bounds one can analyze a simple non-private algorithm with accuracy $1/3$ that has query complexity $O(d \log d)$, which is sublinear in the problem size $m$. However, it is not clear whether $O(m)$ queries are necessary to achieve both DP and accuracy in our model. One might be tempted to directly apply the result of~\cite{BunUV18} to say that $\tilde{\Omega}(m)$ queries are necessary, but this does not work as the results of~\cite{BunUV18} focus on sample complexity. In particular, in our model it would be possible to distribute attribute queries across all rows (making $o(d)$ attribute queries in each row) so that every row is (partially) examined but the total number of queries is still $o(m)$. In particular, a sublinear time algorithm can substantially reduce $\ell_1$ and $\ell_2$ sensitivity by ensuring that the maximum number of queries in each row is $o(d)$. \footnote{ Suppose for example, that $m=d^2$ so that there are $d$ attribute columns and $d$ rows and consider two sublinear time algorithms: Algorithm 1 examines the first $\sqrt{d}$ rows and outputs the marginals for these samples. By contrast, Algorithm 2 uses rows $i\sqrt{d}+1$ to $(i+1)\sqrt{d}$  to compute the marginals columns  $i\sqrt{d}+1$ to $(i+1)\sqrt{d}$ for each $i < \sqrt{d}$. Both algorithms examine the same number of cells in the database $d\sqrt{d}$, but the $\ell_1$ (resp. $\ell_2$) sensitivity of the algorithms are quite different. Algorithm 1 has $\ell_1$ (resp. $\ell_2$) sensitivity $\sqrt{d}$ (resp. $1$) while Algorithm 2 has $\ell_1$ (resp. $\ell_2$) sensitivity $1$ (resp. $d^{-0.25}$). }  

The sample complexity lower bound of~\cite{BunUV18} uses fingerprinting codes to show that the output of an algorithm that is both DP and reasonably accurate for the one-way marginal problem can be used to reidentify an individual record, which contradicts the DP property. Intuitively, fingerprinting codes provide the guarantee that if an algorithm obtains an accurate answer after examining {\em at most} $c$ rows in the database then it is possible to reidentify at least one of the corresponding users. However, if the attacker examines more than $c$ rows, then \emph{we cannot prove that privacy is violated as fingerprinting codes no longer provide the guarantee that we can reidentify one of the corresponding users}. In our model, an algorithm is allowed to make arbitrary attribute queries and is not restricted to querying all attributes corresponding to a fixed row, thus it is more difficult to prove a lower bound of this nature in our model. In particular, instead of sampling $c$ rows an attacker could distribute the attribute queries across all rows (making $c=o(d)$ attribute queries in each row). The total number of cells examined is still $cd$, but the overall coalition has size $d \geq c$. Fingerprinting codes provide no guarantee of being able to trace a colluder since the overall coalition (number of rows in which some query was made) is larger than $c$. Thus, we cannot prove that privacy is violated.

%The parameter $c$ (coalition size) depends on the particular fingerprinting codes. If the  

%if the answer was generated after examining at most $c$ rows in the database then there 

Crucially, their construction relies on the algorithm being able to query the entire row (aka record) of the database and the fact that for a fixed coalition size $c$ a fingerprinting code can trace an individual in any coalition of size $\leq c$ with high probability, as long as the individual \emph{actively} colluded. 

%Even supposing that we restrict the attacker so that he has to query all of the rows in a cell (or none of them) there are still challenges to address in our model....introduce/motivate Pad/Permute trick.

%\subsection{Our Contributions}
{\bfseries Our Contribution. } We give the first separation between the query complexity of a non-private, sublinear-time algorithm and a DP sublinear-time algorithm (up to a log factor). We remind the reader that a lower bound on query complexity naturally implies a lower bound on time complexity as the time taken by an algorithm must be at least the number of queries made. Thus our theorem on query complexity also gives a lower bound for sublinear-time DP algorithms. Recall that our problem size is $m=n \cdot d$ in the following result. 

\begin{theorem}[Informal Theorem]\label{thm:main-inf}
\label{thm:informal}
There exists a problem $\cP$ of size $m$ such that
\begin{enumerate}
    \item $\cP$ can be solved privately with $O(m)$ query and time complexity.
    \item $\cP$ can be solved non-privately with $O(m^{2/3}\log(m)) \in o(m)$ query and time complexity.\label{item:thm-informal-nonprivate}
    \item Any algorithm that solves $\cP$ with $(1/3,1/3)$-accuracy and $(O(1),o(1/n))$-DP must have $\Omega(m / \log(m)) = \Tilde{\Omega}(m)$ query and time complexity. \label{item:thm-informal-sub-dp}
\end{enumerate}
\end{theorem}
We note that \cite{BunUV18} implies that any accurate, DP algorithm for $\cP$ requires $n \in \Omega(\sqrt{d})$, and we can in fact invoke the theorem for the \emph{hardest} case and choose $m$ so that $n = \Theta(\sqrt{d})$.\footnote{We call $n = \Theta(\sqrt{d})$ the hardest case because, when $n$ becomes larger as a function of $d$, \cite{balle2018privacy} show that subsampling can improve the privacy/accuracy trade-off of existing DP algorithms.} For full details on the definition of the problem and the formal version of the theorem, see \cref{def:ss-our} and \cref{thm:main-ss}. Summarized in words, $\cP$ is solvable under  differential privacy, and there exists a non-private sublinear-time algorithm to solve $\cP$, but any DP algorithm must read (almost) the entire dataset, and thus have at least (nearly) linear running time. Our techniques build upon a rich literature of using fingerprinting codes in DP lower bounds. We note that the $\log(m)$ factor in our main result (Item~\ref{item:thm-informal-sub-dp}) of \cref{thm:main-inf} arises from the nearly-optimal Tardos fingerprinting code used in our lower bound construction. Thus it seems unlikely that this result can be improved unless one bypasses using fingerprinting codes entirely in the DP lower bound construction.

\subsection{Technical Overview}

{\bfseries Fingerprinting code (FPC) and DP. }We start the construction of our lower bound with the privacy lower bounds based on fingerprinting codes \cite{BunUV18}. For a set of $n$ users and a parameter $c \leq n$, an $(n,d,c)$-FPC consists of two algorithms $(Gen, Trace)$. The algorithm $Gen$ on input $n$ outputs a codebook $C\in \{0,1\}^{n\times d}$ where each row is a codeword of user $i \in [n]$ with code length $d = d(n, c)$. It guarantees that if at most $c$ users collude to combine their codewords into a new codeword $c'$ and the new code satisfies some (mild) \emph{marking condition} --- namely that if every colluder has the same bit $b$ in the $j$-th bit of their codeword, then the $j$-th bit of $c'$ must also be $b$ --- then the $Trace$ algorithm of the fingerprinting code can identify at least one colluder with high probability.\footnote{The idea of fingerprinting codes becomes colorful when imagining a publisher who distributes advance copies to press and wants to add watermarks that are robust, e.g., against pirated copies that result from averaging the copies of multiple colluders. To see that the marking assumption is a mild condition, consider that the codeword is hidden in the much larger content.} Bun et al.~\cite{BunUV18} launched a reidentification attack using fingerprinting codes to show that the output of any accurate algorithm for the one-way marginal problem must satisfy the marking condition for sufficiently large $d$, and therefore, at least one row (i.e., individual) from the database is identifiable --- making this algorithm not private. In more detail, given a fingerprinting code $(Gen,Trace)$, suppose a coalition of $n$ users builds dataset $D \in \{0,1\}^{n \times d}$ where each row corresponds to a codeword of length $d$ from the codebook $Gen$. For $j\in [d]$, if every user has bit $b$ in the $j$-th bit of their codeword then the one-way marginal answer for that column will be $b$. It is shown that any algorithm that has non-trivial accuracy for answering the one-way marginals on $D$ can be used to obtain a codeword that satisfies the marking condition. Therefore, using $Trace$ on such a codeword leads to identifying an individual in dataset $D$. Since an adversary is able to identify a user in $D$ based on the answer given by the algorithm, this clearly violates DP.  

{\bfseries Techniques of~\cite{BunUV18} do not directly apply. } 
In our model, an algorithm only sees a subset of the entries in the entire database via attribute queries. Suppose a coalition of $c \leq n$ users belongs to a dataset $D \in \{0,1\}^{n \times d}$ where each row corresponds to a codeword of length $d$. As a warm-up, let us first assume that an algorithm that solves the one-way marginal problem on input $D \in \{0,1\}^{n \times d}$ always queries for entire rows and that an adversary can simulate a query oracle to the algorithm's queries, i.e., respond with rows that exactly correspond to the set of $c$ colluders (for more details see~\cref{sec:ro}). To apply a fingerprinting code argument to such an algorithm, an adversary must identify a row from this subset of $c$ rows by examining the output of the algorithm. However, since the accuracy guarantee of the algorithm applies (only) to the \emph{whole} dataset, we cannot make the same argument as above to conclude that the marking condition holds for the \emph{subsample} of rows. %Therefore, we need to ensure that any attempt of the algorithm to spoil the marking condition would contradict its accuracy guarantees. 
In other words, we need to ensure that the output of an accurate algorithm that only sees a subsample of rows can also satisfy the marking condition. The techniques of~\cite{BunUV18} do not ensure such a property.  

{\bfseries Permute Rows and Pad and Permute Columns Fingerprinting codes (PR-PPC FPC). } In order to achieve the property described above, we need to ensure that any attempt of the algorithm to spoil the marking condition would contradict its accuracy guarantees. We achieve this property by padding $O(d)$ additional columns to the codebook $C$ to obtain $C' \in \{0,1\}^{n \times d'}$, where $d'=O(d)$, so that (codebook) columns whose output could be modified to violate the marking condition and (padded) columns whose modification would violate the accuracy guarantee are indistinguishable in the subsample with good probability. Padded columns have been used to define a variant of a fingerprinting code in previous work to achieve smooth DP lower bounds~\cite{PTU23}. In our work, we not only need a variant of FPC with padded columns, but we also need to permute the rows of the codebook (see \cref{sec:prppc} for the construction). This is because we need to define a sampling procedure with certain properties for the adversary to obtain a dataset on $c$ rows from a distribution over databases of $n$ rows, and one way to do so is permuting the rows of the codebook and outputting the first $c$ rows (e.g., see~\cref{thm:main-pro-reid} and \cref{thm:main-pro-reid-ss}).

\begin{remark}
    We note that~\cite{BunUV18} used a similar padding technique to argue about obtaining error-robust codes from ``weakly-robust'' codes (see Lemma 6.4 in \cite{BunUV18}). In particular, we could argue that the property that we need for our purpose is achieved by an error robust code. However, we choose to start with a weaker construction, as we do not inherently need the error robustness property. % and using this variant may add unnecessary limitations to our construction. 
\end{remark}

{\bfseries Secret Sharing Encoding. }Finally, we overcome the assumption that the algorithm queries for entire rows by applying a secret sharing scheme to the padded codebook (see \cref{sec:ss}). In particular, an adversary can encode each row $x_i \in \{0,1\}^{d'}$ with respect to a random polynomial of degree $2d'-1$ as a share of size $2d$. The shares are defined by the $d$ codebook values and $d$ random values from the field. For each query of the algorithm, the adversary answers with a share from the second half. Information theory implies that the algorithm can only recover the $d$ codebook values after querying for all $d$ random value shares. Thus, we obtain a derivate of the one-way marginal problem that requires the algorithm to query an entire row to reveal the padded code book row. While there exist a DP algorithm (\cref{thm:dp-pss}) and a sublinear-time algorithm (\cref{thm:subtime-pss}) for this derived problem as well, we show that there exists no sublinear-time DP algorithm (up to a log factor) that can query for arbitrary entries in the database.
%end of related work
%\end prelims

\section{Related Work}\label{sec:related}
Fingerprinting codes were first introduced in the context of DP lower bounds by~\cite{BunUV18}. Prior to their work, traitor-tracing schemes (which can be thought of as a cryptographic analogue of information-theoretic fingerprinting codes) were used by~\cite{DworkNRRV09,Ullman13} to obtain computational hardness results for DP. Subsequent works have refined and generalized the connection between DP and fingerprinting codes in many ways~\cite{SteinkeU16,SteinkeU17,DworkSSUV15,BunSU19,KMS22,KamathLSU19,CWZ215}. The fingerprinting code techniques of proving DP lower bounds have been used in many settings including principal component analysis~\cite{DworkTTZ14}, empirical risk minimization~\cite{BassilyST14}, mean estimation~\cite{BunUV18,KamathLSU19}, regression~\cite{CWZ215}, gaussian covariance estimation~\cite{KMS22}. Recently~\cite{PTU23} use fingerprinting codes to give smooth lower bounds for many problems including the 1-cluster problem and $k$-means clustering.

In the streaming model,~\cite{DinurSWZ23} give a separation between the space complexity of differentially private algorithms and non-private algorithms -- under cryptographic assumptions they show that there exists a problem that requires exponentially more space to be solved efficiently by a DP algorithm vs a non-private algorithm. By contrast, our focus is on lower bounding the running time (query complexity) of a differentially private algorithm. Our bounds do not require any cryptographic assumptions. ~\cite{JRSS21} give a lower bound in the continual release model, in particular they show that there exists a problem for which any DP continual release algorithm has error $\tilde{\Omega}(T^{1/3})$ times larger than the error of a DP algorithm in the static setting where $T$ is the length of the stream.

\section{Preliminaries}\label{sec:prelims}
% \begin{definition}
%     Given a dataset $D=\{x_1,\ldots,x_n\}$ where $x_i \in \{0,1\}^d$, the $j$-th marginal is given by $M_j:=\frac{1}{n}\sum^n_{i=1}x_{i,j}$. The 1-way marginals problem is defined as follows: Given dataset $D \in (\{0,1\}^d)^n$ as input, output all $d$ 1-way marginals of $D$, i.e., output $(M_1, \ldots, M_d)$. 
% \end{definition}
We define a database $D \in \cX^n$ to be an ordered tuple of $n$ rows $(x_1,\ldots,x_n) \in \cX$, where $\cX$ is the data universe. For our purposes, we typically take $\cX = \{0,1\}^d$. Databases $D$ and $D'$ are neighboring if they differ by a single row and we denote this by $D \sim D'$. In more detail, we can replace the $i$-th row of a database $D$ with some fixed element of $\cX$ to obtain dataset $D_{-i} \sim D$. Importantly both $D$ and $D_{-i}$ are databases of the same size.   
\begin{definition}[Differential Privacy~\cite{DworkMNS06}]
Randomized algorithm $\cA: \cX^n \to \cR$ is $(\eps,\delta)$-differentially private if for every two neighboring databases $D \sim D'$ and every subset $S \subseteq \cR$, 
$$\Pr[\cA(D) \in S ] \leq e^\eps \Pr[\cA(D') \in S]+\delta$$
\end{definition}

\begin{definition}[Accuracy]
    Randomized algorithm $\cA: \cX^n \to \bbR^{d}$ is $(\alpha,p)$-accurate for problem $\cP$ if for every $D \in \cX^n$, with probability at least $1-p$, the output of $\cA$ is a vector $a \in \{0,1\}^d$ that satisfies $\vert a_\cP(D) - a \vert \leq \alpha$ where $a_\cP(D)$ denotes the exact solution of the problem $\cP$ on input $D$.
\end{definition}

The following definition of fingerprinting codes is ``fully-collusion-resilient''. For any coalition of users $S$ who collectively produce a string $y \in \{0,1\}^d$ as output, as long as $y$ satisfies the \textit{marking condition} -- for all positions $1 \leq j \leq d$, if the values $x_{ij}$ for all users $i$ in coalition $S$ agree with some letter $s \in \{0,1\}$, then $y_j =s$ -- then the combined codeword $y$ can be traced back to a user in the coalition.  Formally, for a codebook $C \in \{0,1\}^{n \times d}$, and a coalition $S \subseteq [n]$, we define the set of feasible codewords for $C_S$ to be
\[F(C_S) = \{c' \in \{0,1\}^d \mid \forall j \in [d],\exists i \in S, c'_j = c_{ij}\}\]
\begin{definition}[Fingerprinting Codes~\cite{BunUV18}]\label{def:fpc}
    For any $n,d \in \bbN, \xi \in (0,1]$, a pair of algorithms $(Gen,Trace)$ is an $(n,d,c)$-fingerprinting code with security $\xi$ against a coalition of size $c$ if $Gen$ outputs a codebook $C \in \{0,1\}^{n\times d}$ and secret state $st$ and for every (possibly randomized) adversary $\cA_{FP}$, and every coalition $S \subseteq [n]$ such that $\vert S \vert \leq c$, if we set $c' \leftarrow_R \cA_{FP}(C_S)$ then 
    \begin{enumerate}
        \item $\Pr[c' \in F(C_S) \wedge Trace(c') = \perp ]\leq \xi $
        \item $\Pr[Trace(c') \in [n]\setminus S]\leq \xi$
    \end{enumerate}
    where $C_S$ contains the rows of $C$ given by $S$, and the probability is taken over the coins of $C$, $Trace$, and $\cA_{FP}$. The algorithms $Gen$ and $Trace$ may share a common state denoted as $st$.
\end{definition}
\begin{remark}
    Although the adversary $\cA_{FP}$ is defined as taking the coalition of users' rows as input, we may abuse this notation and consider the entire codebook or a different input (related to the codebook) altogether. This does not change the security guarantees of the FPC against adversary $\cA_{FP}$ because the security guarantee holds as long as the output of $\cA_{FP}$ is a result of the users in the coalition $S$ actively colluding. 
\end{remark}
\begin{theorem}[Tardos Fingerprinting Code~\cite{Tardos08}]\label{thm:tardos}
    For every $n\in \bbN$ and $4\leq c\leq n$, there exists an $(n,d,c)$-fingerprinting code of length $d=O(c^2 \log(n/\xi))$ with security $\xi \in [0,1]$ against coalitions of size $c$.
\end{theorem}
% \begin{theorem}
%     For every $1 \leq k \leq n$ there is a $k$-collusion resilient fingerprinting code of length $d=O(k^2 \log(n))$ for $n$ users with failure probabilty $\gamma=\frac{1}{n^2}$ and an efficiently computable Trace function. 
% \end{theorem}

\begin{theorem}[Gaussian Mechanism, \cite{dwork2014algorithmic}]
    Let $\eps, \delta \in (0,1)$ and $f : \mathbb{N}^d \rightarrow \mathbb{R}^d$. For $c > \sqrt{2 \ln(1.25 / \delta)} / \eps$, the Gaussian Mechanism with standard deviation parameter $\sigma \geq c \Delta_2 f$ is $(\eps, \delta)$-DP, where $\Delta_2$ is the $\ell_2$-norm sensitivity of $f$.
\end{theorem}

\begin{lemma}\label{lem:gauss}
 For $n \geq \sqrt{200 d \ln(20d) \ln (1.25 / \delta)} / \eps \in \Tilde{\Omega}(\sqrt{d})$, given a dataset $D \in \{0,1\}^{n \times d}$, there exists a $(1/10,1/10)$-accurate $(\eps,\delta)$-DP algorithm that solves the one-way marginals problem with $O(m)$ attribute queries, where $m=n \cdot d$.
\end{lemma}
\begin{proof}
We note that the $\ell_2$-sensitivity of the one-way marginals problem on a database $\{0,1\}^{n \times d}$ is $\sqrt{d}/n$. For $n \geq \sqrt{200 d \ln(20d) \ln (1.25 / \delta)} / \eps \in \Tilde{\Omega}(\sqrt{d})$, the Gaussian Mechanism is $(1/10, 1/10)$-accurate with
\begin{align*}
    &\sigma = \frac{\sqrt{2 \ln(1.25 / \delta)}}{\eps} \cdot \frac{\eps \sqrt{200d}}{\sqrt{2 \ln (1.25 / \delta) \cdot d \ln(10 d)}} = \frac{1}{\sqrt{200 \ln (10 d)}} ,\\&\text{as }
    \Pr_{X \sim \mathcal{N}(0, \sigma^2)} \left[ X \geq \frac{1}{10} \right] \leq 2e^{-\frac{1}{200\sigma^2}} \leq \frac{1}{20d}
\end{align*}
by the Cramer-Chernoff inequality and a union bound over all $d$ columns of the dataset.
\end{proof}

\section{Permute Rows and Pad and Permute Columns Fingerprinting Codes (PR-PPC FPC)}\label{sec:prppc}
In this section, we first introduce our pad and permute variant of the original fingerprinting codes where we Permute Rows and Pad and Permute Columns (PR-PPC $(n,d,c,\ell)$-FPC). Given $(Gen, Trace)$ of an $(n,d,c)$-FPC, we construct $Gen'$ and $Trace'$ in \cref{alg:gen'} and \cref{alg:trace'} to produce a PR-PPC $(n,d,c,\ell)$-FPC. In more detail, $Gen'$ samples codebook $C$ and secret state $st$ from $Gen$ where $C\in \{0,1\}^{n\times d}$. It then permutes the rows via a random permutation $\pi_R$, after which it pads $2\ell$ columns and performs another random permutation $\pi$ on the columns. Then, it releases the resulting codebook $C' \in \{0,1\}^{n\times d'}$, where $d'=d+2\ell$ and $\ell$ is the parameter which controls the number of columns padded to $C$. Note that the row permutation $\pi_R$ is public while the column permutation $\pi$ is part of the new secret state $st'$. The algorithm $Trace'$ receives an answer vector of dimension $d'$ and uses its secret state $st'=(st,\pi)$ to feed the vector entries that correspond to the original first $d$ columns via $\pi^{-1}$ to $Trace$ and releases the output of $Trace$. We obtain the following result directly from the definition of $Gen'$ and $Trace'$. 
\begin{corollary}\label{cor:rpppc}
    Given an $(n,d,c)$-FPC, $\ell \geq 0$, and the corresponding PR-PPC $(n,d,c,\ell)$-FPC, the properties of $Trace$ as stated by \cref{def:fpc} translate directly to $Trace'$.
\end{corollary}
% \begin{algorithm}[!htb]
% \caption{$\textsf{PADPERM}$}
% \label{alg:padperm}
% \begin{algorithmic}[1]
% \Require{Number of 0/1 columns $\ell$ of dimension $n \times 1$ to be padded, Permutation $\pi$, Input matrix $C \in \{0,1\}^{n \times d}$.}
% \State $C_{pad}\leftarrow $ Pad $\ell$ columns of all 1's and $\ell$ columns of all 0's to matrix $C$ 
% \State Output $C_{pp} \leftarrow $Permute the columns of $C_{pad}$ according to random permutation $\pi$
% \end{algorithmic}
% \end{algorithm}

\begin{algorithm}[!htb]
\caption{$Gen'$}
\label{alg:gen'}
\begin{algorithmic}[1]
\Require{Number of users $n \in \bbN$, number of padded 0/1 columns $\ell$}
\State Sample codebook $(C,st) \leftarrow Gen(n)$ such that $C \in \{0,1\}^{n \times d}$. 
\State Sample random permutation $\pi:[d'] \to [d']$ where $d':= d+2 \ell$. For an $n \times d'$ matrix $A$, define $n\times d'$ matrix $\pi(A)$ such that $\pi(j)$-th column of $\pi(A)$ equals to the $j$-th column of $A$ for every $j \in [d']$.
\State Sample random permutation $\pi_R : [n] \to [n]$. For an $n$-row matrix $A$, define $\pi_R(A)$ such that $\pi_R(i)$-th row of $\pi_R(A)$ equals to the $i$-th row of $A$ for every $i \in [n]$. 
\State $C^{\pi_R} \leftarrow $ Permute rows of $C$ via random permutation $\pi_R$.
\State $C_{pad}\leftarrow $ Pad $\ell$ columns of all 1's and $\ell$ columns of all 0's to matrix $C^{\pi_R}$.
\State $C' \leftarrow $Permute the columns of $C_{pad}$ according to random permutation $\pi$. %\Comment{$C' \in \{0,1\}^{n \times d'}$ where $d':= d+2 \ell$}
%\State $C'\leftarrow \textsf{PADPERM}(\pi,C^{\pi_R},\ell)$ \Comment{$C' \in \{0,1\}^{n \times d'}$ where $d':= d+2 \ell$}
\State Output $C'$ along with the new secret state $st':=(st, \pi)$ and permutation $\pi_R$.
\end{algorithmic}
\end{algorithm}

\begin{algorithm}[!htb]
\caption{$Trace'$}
\label{alg:trace'}
\begin{algorithmic}[1]
\Require{Answer vector $\mathbf{a} =(a^1, \ldots, a^{d'}) \in \{0,1\}^{d'}$, secret state $st'=(\pi,st)$ }
\State Output $Trace(\mathbf{a}_{og},st)$ where $\mathbf{a}_{og}=(a^{\pi(1)}, \ldots, a^{\pi(d)}) \in \{0,1\}^{d}$.
\end{algorithmic}
\end{algorithm}

We define the \emph{feasible sample property} of an FPC below. Informally, it states that if we have an algorithm that takes a sample (or subset) of rows from a codebook as input, and the algorithm's output is a feasible codeword for the entire codebook, then the same output should be a feasible codeword for the sample.

\begin{definition}[Feasible Sample Property]\label{def:feasible-sample}
    Let $C \in \{0,1\}^{n\times d}$ be a codebook of an $(n,d,c)$-FPC, $S \subseteq [n]$ be a coalition and $C_S \subseteq C$ be the matrix consisting of the corresponding rows indexed by $S$. Given an algorithm $\cA$ that takes as input $C_S$ and outputs a vector $\mathbf{o} \in \{0,1\}^d$, the feasible sample property states that if $\mathbf{o} \in F(C)$, then $\mathbf{o} \in F(C_S)$. 
\end{definition}

\begin{lemma}\label{lem:win-sample}
PR-PPC $(n,d,c,\ell)$-FPC satisfies the feasible sample property with probability at least $1-\frac{d}{\ell}$.
\end{lemma}

\begin{proof}
Given $(Gen',Trace')$ of PR-PPC $(n,d,c,\ell)$-FPC which produces codebook $C' \in \{0,1\}^{n \times d'}$ and sampling algorithm $\cA$ which takes as input $C'_S \subseteq C'$ and outputs a vector $\mathbf{o} \in \{0,1\}^{d'}$ where $d'=d+2\ell$, we define the event $\textsf{BAD}_S$ as ${\textbf{o}} \in F(C')$ but ${\textbf{o}} \not\in F(C'_S)$.

We denote the indices of columns of $C'$ in which all the bits are 1 as $C'_{\vert1}$ and the indices of columns in which all the bits are 0 as $C'_{\vert0}$. Similarly, we define $C'_{S \vert 1}$ and $C'_{S \vert 0}$ for the columns that are all 1 and all 0 in $C'_S$, respectively. Note that $C'_{\vert 1} \subseteq C'_{S \vert 1}$ and $C'_{\vert 0} \subseteq C'_{S \vert 0}$. Since by definition, algorithm $\cA$ only has access to the set of rows in $C'_S$, in order for the output ${\textbf{o}}$ to satisfy ${\textbf{o}} \in F(C')$ but ${\textbf{o}} \not\in F(C'_S)$, an adversary that aims for $\textsf{BAD}_S$ must flip a bit of the resulting codeword that originates from $\{C'_{S \vert1} \cup C'_{S \vert0} \} \setminus \{C'_{\vert1} \cup C'_{\vert0} \}$. In other words, the event $\textsf{BAD}_S$ occurs only if the adversary identifies a column from $C'$ that contains at least one 0 and one 1, but reduces to an all-1 or all-0 column in $C'_S$. %Note that if the adversary flips the bits resulting from the padded consensus columns, then the resulting answer will not be a feasible codeword for the entire codebook either. 

More formally, the adversary can pick the bit $b \in \{0,1\}$ resulting in a column from $C'_{S \vert b}$ to flip. The probability that the adversary correctly identifies a column from $C'_{S \vert b} \setminus C'_{\vert b}$ is at most $\frac{\lvert C'_{S \vert b} \setminus C'_{\vert b} \rvert}{\lvert C'_{S \vert b} \rvert}$. Observe that $\lvert C'_{S \vert b} \rvert \geq \lvert C'_{\vert b} \rvert \geq \ell$ due to the $\ell$ padded all-$b$ columns, and therefore $\lvert C'_{S \vert b} \setminus C'_{\vert b} \rvert \leq \lvert C'_{S \vert b} \rvert - \lvert C'_{\vert b} \rvert \leq (d + \ell) - \ell = d$. Thus, the probability that event $\textsf{BAD}_S$ occurs is at most $\max_{b \in \{0,1\}} \frac{\lvert C'_{S \vert b} \setminus C'_{\vert b} \rvert}{\lvert C'_{S \vert b} \rvert} \leq \frac{d}{\ell}$. 
\end{proof}

\section{Lower Bound}\label{sec:lowerbound}
We present our lower bound for sublinear-time DP algorithms in this section. The main idea behind our lower bound proof is to construct a reidentification attack in which the adversary $\mathcal{B}$ is given oracle access to the algorithm $\mathcal{A}$ that accurately solves our proposed problem $\mathcal{P}$. Using fingerprinting codes, we will show that the adversary can use the output of $\mathcal{A}$ to reidentify a subset of the input set given to $\mathcal{A}$ with high probability. We invoke existing fingerprinting code bounds to achieve our final lower bound result. 
In \cref{sec:ro} we discuss a  warm-up problem, where the algorithm can only make row queries to release the one-way marginals of the dataset. We present our main result and the lower bound construction for algorithms that can make arbitrary attribute queries in \cref{sec:ss}. In the sequel, our problem space has size $m=n\cdot d = \Omega(d \sqrt{d})$ and our results will be in terms of dimension $d$.

\subsection{Warm Up: Using a Random Oracle}\label{sec:ro}
In this section, we first present a warm-up problem which we call the Random Oracle Problem ($\cP_{RO}$). This is an extension of the one-way marginals problem in the following manner --- for an input dataset $D = (x_1, \ldots, x_n)$, and access to a random oracle $H$, the $\cP_{RO}$ problem takes as input an encoded dataset $D_H=(z_1, \ldots, z_n)$ in which $z_i =H(i) \oplus x_i$, and outputs the one-way marginals of the underlying dataset $D$ (see \cref{def:ro} for a formal definition). The main intuition for introducing such a problem is that we want to force an algorithm that solves this problem to query an entire row. Recall that in our (final) model an algorithm is allowed to make arbitrary attribute queries. We shall see how even in this simpler formulation (where the algorithm is forced to query entire rows instead of attributes), we need to use the variant of fingerprinting codes PR-PPC introduced in \cref{sec:prppc}, to prove a DP lower bound. The intuition behind an algorithm solving $\cP_{RO}$ having to query entire rows is the following --- Given $D_H$, in order to approximate or exactly compute the one-way marginals of $D$, an algorithm needs to query $H(i)$ for $i \in [n]$, as otherwise, by the properties of the random oracle and one-time pad (OTP), the value of $x_i$ is information-theoretically hidden.

\begin{definition}[Random Oracle Problem $\cP_{RO}$]\label{def:ro}
Given a random oracle $H: [n] \to \{0,1\}^d$, and a dataset $D=(x_1, \ldots, x_n)$ where $x_i \in \{0,1\}^d$,
define dataset $D_{H}:=(z_1, \ldots, z_n)$ where $z_i=H(i)\oplus x_i$. For simplicity of notation, we refer to the operation for obtaining $D_H$ from $D$ as $H(D)$. The problem $\cP_{RO}$ on input $D_H$ releases the one-way marginals of $D$.

We use $\cP_{RO}(D)$ to denote that $\cP_{RO}$ releases the one-way marginals of the underlying dataset $D$. 
\end{definition}

\subparagraph*{Query Model.} On input $D_H \in (\{0,1\}^d)^n$, an algorithm can query the random oracle $H$ through \emph{row queries}, i.e., given a row index $i\in [n]$ of $D_H$, the answer given is $H(i) \in \{0,1\}^d$. We note that our final result in this subsection will still be presented in the form of attribute queries as 1 row query translates to $d$ attribute queries.

Observe that there exists an $(\eps,\delta)$-DP algorithm for $\cP_{RO}(D)$ that on input $D_H$, queries the entire dataset via row queries to the random oracle $H$, i.e., it makes $dn = O(d \sqrt{d})$ queries. After obtaining the rows to the underlying dataset $D$ it releases the one-way marginals using the Gaussian Mechanism (see \cref{lem:gauss}). We also note that there exists a sublinear non-DP algorithm for $\cP_{RO}(D)$ which makes $O(d \log d)$ queries, which is a simple corollary of Hoeffding bounds. Our goal in this section is to prove the lower bound below. Recall that $n \in \Omega(\sqrt{d})$, so the problem size is $\Omega(d \sqrt{d})$.

\begin{theorem} [Lower Bound for $\cP_{RO}$]\label{thm:main-ro-inf}
Any algorithm that solves the problem $\cP_{RO}$ with $s$ attribute query complexity, $(1/3,1/3)$-accuracy and $(O(1),o(1/s))$-DP must have $s=\Omega(d\sqrt{d/\log(d)})$.
%Any sampling algorithm that solves the problem $\cP_{RO}$ with $(1/3,1/3)$-accuracy and $(O(1),o(1/c))$-DP must have $c=\Omega(d^2\sqrt{d/\log(d)})$ query complexity. 
   % For every $(\eps,\delta)$-DP algorithm that solves problem $\cP_{RO}$ with query complexity $s$, it holds that $s = \Omega(\sqrt{d/\log(n)})$. 
\end{theorem}

We present a high level overview of the proof of \cref{thm:main-ro-inf} here. We first show that given an $(n,d,c)$-FPC, there exists a distribution on $c$ rows from which an adversary $\cB$ can sample and create an $n$-row input instance for an algorithm $\cA$ that accurately solves $\cP_{RO}$ (see \cref{thm:main-pro-reid}). Next we argue that the rounded output of $\cA$, denoted as $\mathbf{a}$, is a feasible codeword for the sample of $c$ rows as long as $\cA$ is accurate in a non-trivial manner and $\mathbf{a}$ is feasible for the entire dataset (see \cref{clm:success-A}). The adversary $\cB$ can then use the output from $\cA$ to (potentially) reidentify an individual from the coalition of size $c$. Next we relate these claims back to DP through \cref{lem:acc-dp}, which states that if there exists a distribution $\cC$ on $c \leq n$ row databases according to \cref{thm:main-pro-reid}, then there is no $(\eps,\delta)$-DP algorithm $\cA$ that is $(1/3,1/3)$-accurate for $\cP_{RO}$ with $\eps=O(1)$ and $\delta=o(1/c)$. Finally, invoking the Tardos construction for fingerprinting codes in \cref{thm:tardos} gives us our lower bound.

\begin{theorem}\label{thm:main-pro-reid}%
 For every $n,d \in \bbN$, $\xi \in [0,1]$ and $c \leq n$, if there exists an $(n,d,c)$-fingerprinting code with security $\xi$, then there exists a distribution on $c$-row databases $\cC_\cS$, a row permutation $\pi_R:[n] \to [n]$, and an adversary $\cB$ for every randomized algorithm $\cA$ with row query complexity $c$ and $(1/3,1/3)$-accuracy for $\cP_{RO}$ such that 
   \begin{enumerate}
       \item \label{it:reid-sec}$\Pr_{C'_{S }\leftarrow \cC_\cS} [\cB^{\cA}(C'_S ) = \perp] \leq \xi$       
       \item \label{it:reid-sound} For every $i \in [c]$, $\Pr_{C'_S\leftarrow \cC_\cS}[\cB^{\cA}(C'_{S_{-i}}) = \pi_R^{-1}(i)] \leq \xi$.% where $\pi_R:[n] \to [n]$ is a public permutation.
   \end{enumerate}
   The probabilities are taken over the random coins of $\cB$ and the choice of $C'_S$.   
   
%   with probability at least $1-\xi$, $\cB$ can successfully re-identify a user in a coalition of size $c$ defined by $D_c \leftarrow \cD_c$.    
   %a distribution on $c$-row databases $D_c \in (\{0,1\}^d)^c$ that is $(\xi,\xi)$-reidentifiable from $(1/3,0, 1/3)$-accurate answers to $\cP_{RO}$. 
\end{theorem}
 Let $(Gen, Trace)$ be the promised $(n,d,c)$-fingerprinting code in the theorem statement. We first construct a PR-PPC $(n,d,c,\ell)$-FPC with $\ell:=100d$ (see \cref{sec:prppc} for details). 
  
 The distribution $\cC_\cS$ on $c$-row databases is implicitly defined through the sampling process below
 \begin{enumerate}
 \item Let $C' \leftarrow Gen'(n,100d)$ (see \cref{alg:gen'}) where $C' \in \{0,1\}^{n \times d'}$ and $d'=d+100d=101d$. Note that $Gen'$ also outputs $\pi_R$ which is a public permutation on rows. 
     \item Let $C'_S =(x_1, \ldots, x_c)\in \{0,1\}^{c \times d'}$ be the first $c$ rows of $C' \in \{0,1\}^{n\times d'}$
     \item Output $C'_S$
 \end{enumerate}

Next we define the privacy adversary $\cB$.

    \noindent{\bfseries {Adversary $\cB$ Algorithm. }}
    Adversary $\cB$ receives $C'_S$ as input and does the following: 
    \begin{enumerate}
        \item Create a database $D =(r_1,\ldots, r_n) \in \{0,1\}^{n \times d'}$ where each row $r_i \in \{0,1\}^{d'}$ consists of 0/1 entries sampled independently and uniformly at random.% \enote{need to be more specific: $r_i=0,1$ independently and  u.a.r. }
        \item Given oracle access to randomized algorithm $\cA$ which solves $\cP_{RO}$ on input $D$, $\cB$ simulates the answer to the distinct $i_j$-th row query (where $j\in [c]$) made by $\cA$ to random oracle $H$ as follows: 
        \begin{enumerate}
            \item  Return $H(i_j):= r_{i_j} \oplus x_j$. %\enote{need to say somewhere that we treat the 0/1's as element of the field $F_2$ (when we take oplus) and of the reals. Feeling still a bit uneasy about this abuse...}
        \end{enumerate}
        \item Let $\mathbf{a}$ be the output of $\cA(D)$ where $\mathbf{a} \in [0,1]^{d'}$. Round each entry of $\mathbf{a}$ to $\{0,1\}$, call this new vector $\mathbf{\bar{a}} \in \{0,1\}^{d'}$. %\enote{round to the closest and split arbitrarily if .5?}
         \item Output $Trace'(\mathbf{\bar{a}})$
    \end{enumerate}
%    Note that $\cA$ can make at most $c$ queries to the random oracle $H$ when computing $\mathbf{a}$. Let $i_1,\ldots, i_c$ be the distinct queries made to $H$, then $\cA$ receives $H(i_j):= r_{i_j} \oplus x_j$ as answers. 

    \subparagraph*{Analysis.} We focus on proving that Property \ref{it:reid-sec} and Property \ref{it:reid-sound} of the theorem statement are indeed satisfied by adversary $\cB$.
    
    Recall the notation in \cref{def:ro} where $\cP_{RO}(C')$ means that $\cP_{RO}$ releases the one-way marginals of the underlying dataset $C'$. We first show that $\cA$ solving $\cP_{RO}(H(D))$ is perfectly indistinguishable from $\cA$ solving $\cP_{RO}(C')$ in \cref{clm:indis-pro}. This is necessary as $Trace'$ can only identify an individual in the coalition of size $c$ with respect to the codebook $C'$ produced by $Gen'$.
    \begin{lemma}\label{clm:indis-pro}
    $\cA$ solving $\cP_{RO}(H(D))$ is perfectly indistinguishable from $\cA$ solving $\cP_{RO}(C')$.   
    \end{lemma}
    \begin{proof}
    We define the following experiments. 

    \noindent{\bfseries {Real World.}}
    \begin{enumerate}
        \item Given $C' = (x_1, \ldots, x_n)$ where $(C',st') \leftarrow Gen'(\ell)$ with $\ell= 100d$, let $C'_S= (x_1, \ldots,x_c)$.
        \item Create a database $D =(r_1,\ldots, r_n)$ where $ r_i \in \{0,1\}^{d'}$ are random entries. 
        \item Let $\mathbf{a}$ be the output of $\cA(D)$ where $\mathbf{a} \in [0,1]^{d'}$. Simulate $H$ as follows: 
        \begin{enumerate}
            \item Let $i_1, \ldots, i_c$ be distinct queries made to $H$. For $j\in [c]$, fix $H(i_j):= r_{i_j} \oplus x_j$
        \end{enumerate}
    \end{enumerate}

    \noindent{\bfseries {Ideal World.}}
    \begin{enumerate}
        \item Given codebook $C' = (x_1, \ldots, x_n)$ where $(C',st') \leftarrow Gen'(\ell)$ with $\ell= 100d$, let $H(C') = (z_1, \ldots, z_n)$ (see \cref{def:ro} for $H(\cdot)$ notation). 
        \item Let $\mathbf{a} \leftarrow \cA(H(C'))$ where $\mathbf{a} \in [0,1]^{d'}$. 
        \begin{enumerate}
            \item Let $i_1, \ldots, i_c$ be distinct arbitrary queries made to $H$. For $j \in [c]$, $H$ returns the following answer $H(i_j):= z_{i_j} \oplus x_j$
        \end{enumerate}
    \end{enumerate}
    
    In the \textbf{Real World}, $\cA$ is provided $D=(r_1, \ldots, r_n)$ as input (where $D$ is generated in the same manner as by adversary $\cB$), while the \textbf{Ideal World} is one in which $\cA$ takes $H(C')$ as input. We show that $\cA$ learns the same information in the \textbf{Real World} and the \textbf{Ideal World}, i.e., these views are perfectly indistinguishable. 
Observe that the only difference from the viewpoint of $\cA$ between the \textbf{Real World} and the \textbf{Ideal World} is that $H$ is simulated in the former via indices fixed by $C'_S$ whereas $H$ is queried on arbitrary indices in the latter. Since the rows of $C'$ have already been permuted (recall \cref{alg:gen'}), by nature of the random oracle $H$, these two instances are perfectly indistinguishable.  % \enote{do we need to formalize statistical indistinguishability?}
    \end{proof}

Recall that the security condition of the fingerprinting code (see \cref{def:fpc}) only holds if $\bar{\mathbf{a}}$ is a feasible codeword for the coalition of rows, i.e., $C'_S$ in our case. 
The following lemma states that if $\cA$ is accurate for $\cP_{RO}(C')$, then the rounded output of $\cA$ is indeed a feasible codeword for both $C'$ and $C'_S$.  
\begin{lemma}\label{clm:success-A}
    Suppose $\cA$ is $(1/3,1/3)$-accurate for $\cP_{RO}(C')$. Then the rounded output $\mathbf{\bar{a}}$ from algorithm $\cA$ is a feasible codeword for both $C'$ and $C'_S$ with probability at least $1-\frac{1}{3}-\frac{1}{100}$.      
    
    In other words, with probability at least $1-\frac{1}{3}-\frac{1}{100}$, $\mathbf{\bar{a}} \in F(C')$ and $\mathbf{\bar{a}} \in F(C'_S)$.
\end{lemma}
\begin{proof}
Assuming that $\cA$ is $(1/3,1/3)$-accurate for $\cP_{RO}(C')$, we first show that $\mathbf{\bar{a}}$ is a feasible codeword for $C'$ with probability at least 2/3. By the accuracy guarantee of $\cA$, we know that for any column $i_j$ $\vert \mathbf{a}_{i_j} - a_{i_j} \vert \leq 1/3$ where $a_{i_j}$ is the actual 1-way marginal for column $i_j$ with probability at least 2/3. Thus for any column $i_j$ of all 1's in $C'$, $\mathbf{a}_{i_j} \geq 2/3$ which means $\mathbf{\bar{a}}_{i_j}=1$, thus satisfying the marking condition. A similar argument holds for the case when a column is all 0's.

Next, using the fact that we use a PR-PPC $(n,d,c,100d)$-FPC and that $\mathbf{\bar{a}}\in F(C')$ with probability at least 2/3, we can invoke \cref{lem:win-sample} which states that the feasible sample property is satisfied by our PR-PPC FPC construction. Note that in our case, the sampling algorithm described in \cref{def:feasible-sample} is $\cA$ together with the postprocessing step of rounding the output of $\cA$. Also, even though $\cA$ takes the entire dataset as input, it effectively only has access to the rows of the underlying sample via queries to $\cB$ and thus satisfies the properties required in \cref{def:feasible-sample}. 
\cref{lem:win-sample} states that with probability $\leq \frac{1}{100}$, $\bar{\mathbf{a}}$ is not a feasible codeword for $C'_S$. By a union bound we have that $1-\frac{1}{3}-\frac{1}{100}$, $\mathbf{\bar{a}}$ must be a feasible codeword for $C'_S$. 
\end{proof}

\begin{proof}[Proof of \cref{thm:main-pro-reid}]
From the above \cref{clm:success-A}, we have that $\cA$ is $(1/3,1/3)$-accurate for $\cP_{RO}(C')$ implies that $\bar{\textbf{a}}$ is a feasible codeword for $C'_S$.  By the security of the fingerprinting code, \cref{cor:rpppc,clm:indis-pro}, we have that $\Pr[\bar{\textbf{a}} \in  F(C'_S) \wedge Trace'(\bar{\textbf{a}}) = \perp ]\leq \xi$. Since $\cB$ releases the output of $Trace'(\mathbf{\bar{a}})$, the event $\cB^{\cA}(C'_S) = \perp$ is identical to $Trace'(\bar{\textbf{a}} ) = \perp$. Thus Property \ref{it:reid-sec} of the theorem statement which states that the probability that $\cB$ outputs $\perp$ is bounded by $\xi$ follows. Property \ref{it:reid-sound} follows directly from the soundness property of the fingerprinting code. 
\end{proof}

\begin{lemma}\label{lem:acc-dp}
Suppose there exists a distribution on $c \leq n$ row databases $\cC_\cS$ according to \cref{thm:main-pro-reid}. Then there is no $(\eps,\delta)$-DP algorithm $\cA$ with query complexity $c$ that is $(1/3,1/3)$-accurate for $\cP_{RO}$ with $\eps=O(1)$ and $\delta=o(1/c)$.
\end{lemma}
\begin{proof}
    Suppose $C'_S$ is sampled from the distribution on $c$-row databases $\cC_\cS$ and $\cB$ is the adversary from \cref{thm:main-pro-reid}. From the lemma statement we know that $\cA$ is $(1/3,1/3)$-accurate, thus using \cref{clm:success-A} and \cref{thm:main-pro-reid}, we have that $\Pr[\pi_R(\cB^\cA(C'_S)) \in [c]] \geq 1-\frac{1}{3}-\frac{1}{100}-\xi \geq \Omega(1)$. By an averaging argument, this means that there exists some $i^* \in [c]$ for which $\Pr[\pi_R(\cB^\cA(C'_S))=i^*] \geq \Omega(1/c)$. However, if $\xi = o(1/c)$ by Property \ref{it:reid-sound} in \cref{thm:main-pro-reid} we have that $\Pr[\pi_R(\cB^{\cA}(C'_{S_{-i^*}})) =i^*] \leq \xi = o(1/c)$. 
    
    In other words, the probability of $\cB^\cA$ outputting a fixed output $i^*$ on neighboring input databases $C'_{S}$ and $C'_{S_{-i^*}}$ is different, which violates $(\eps,\delta)$-DP for any $\eps=O(1)$ and $\delta=o(1/c)$. We note here that since $\cA$ can make at most $c$ row queries, the DP guarantee for $\cA$ must hold for any neighboring sample of $c$ rows. Since $\cB$ does some postprocessing of the output from $\cA$, and we have shown that $\cB$ cannot be $(\eps,\delta)$-DP, this implies that $\cA$ cannot be $(\eps,\delta)$-DP for any $\eps=O(1)$ and $\delta=o(1/c)$.   

\end{proof}

\subsection{Using a Secret Sharing Encoding}\label{sec:ss}
In this section, we remove the requirement of an algorithm querying an entire row that we enforced in the previous section. We first define the security requirement of a general encoding scheme that is sufficient to construct our DP lower bound in \cref{def:sec-game}. We then show that the Shamir encoding as defined in~\cref{def:ss-bit} satisfies the security requirement (see~\cref{thm:sec-A}). We define a problem based on this secret sharing encoding called $\cP_{SS}$ (see~\cref{def:ss-our}) that uses the encoding to release the one-way marginals of an underlying dataset. Finally, we show that this problem cannot have a sublinear time DP algorithm with reasonable accuracy (see~\cref{thm:main-ss}). Intuitively, the security guarantee of the secret sharing scheme will force any algorithm that solves $\cP_{SS}$ to query many attributes (per row), thus giving our final lower bound.

\begin{definition}[Security Game] \label{def:sec-game} 
Let $\mathtt{Exp}(\mathsf{Enc}_d,\mathcal{A},q,d,x)$ denote the following experiment: (1) the challenger computes $y_0 \leftarrow \mathsf{Enc}_d(x)$ and $y_1 \leftarrow \mathsf{Enc}_d(0^d)$, picks a random bit $b$ and outputs $y= y_b$. (2) $\mathcal{A}^y(d,q,x)$ is given oracle access to $y$ and may make up to $q$ queries to the string $y$. (3) The game ends when the attacker $\cA$ outputs a guess $b'$. (4) The output of the experiment is $\mathtt{Exp}(\mathsf{Enc}_d,\mathcal{A},q,d,x)=1$ if $b'=b$ and the attacker made at most $q$ queries to $y$; otherwise the output of the experiment is $\mathtt{Exp}(\mathcal{A},q,d,x)=0$. We say that the scheme $\mathsf{Enc}_d$ is  $(q(d),d,\gamma(q,d))$-secure if for all $x \in \{0,1\}^d$ and all attackers $\mathcal{A}$ making at most $q$ queries we have 
\[ \Pr[\mathtt{Exp}(\mathsf{Enc}_d,\mathcal{A},q,d,x)=1] \leq \frac{1}{2} + \gamma(q,d) \]
\end{definition}

\begin{definition} (Shamir Encoding) \label{def:ss-bit}
 Given a row $x_i \in \{0,1\}^d$ where $i\in [n]$ and a field $\bbF$ s.t. $|\bbF| > {4d}$ let $SS_d(x_i)$ be the following encoding (1) pick random field elements $\alpha^{(i)}_1,\ldots, \alpha^{(i)}_d, \alpha^{(i)}_{d+1},\ldots, \alpha^{(i)}_{3d}$ (distinct) and $z^{(i)}_{d+1},\ldots,z^{(i)}_{2d}$ and define the polynomial $p_i(\cdot)$ of degree $2d-1$ s.t. $p_i(\alpha^{(i)}_j) = x_j$ and $p_i(\alpha^{(i)}_{d+j}) = z^{(i)}_{d+j}$ for $j \leq d$. (2) publish  $SS_d(x_i)=\left(\alpha^{(i)}_1,\ldots,\alpha^{(i)}_d, \{(\alpha^{(i)}_j,p_i(\alpha^{(i)}_j))\}_{j=d+1}^{3d}\right)$ as share of $x_i$. 
 \end{definition}

\begin{definition}[Secret Sharing Problem $\cP_{SS}$] \label{def:ss-our}
Let dataset $D:=(x_1,\ldots,x_n) \in \{0,1\}^{n\times d}$. Given $D_S:= (SS_d(x_1),\ldots, SS_d(x_n))$, the goal of the secret-sharing problem $\cP_{SS}$ is to release all the one-way marginals of dataset $D$.

We use $\cP_{SS,d}(D)$ to denote that $\cP_{SS}$ releases the one-way marginals of the underlying dataset $D$ with dimension $d$. 
\end{definition}

\subparagraph*{Query Model.} On input $D_S$, an algorithm solving the $\cP_{SS}$ problem can make attribute queries to obtain the underlying dataset $D$ and release its one-way marginals. For a row $i \in [n]$, the $i_j$-th attribute query returns the pair of field elements $(\alpha^{(i)}_{j+d},p(\alpha^{(i)}_{j+d}))$ of share $SS_d(x_i)$ for $1 \leq j \leq 2d$. We note that the prefix of $SS_d(x_i)$ given by $\alpha^{(i)}_1,\ldots,\alpha^{(i)}_d$ is published separately after an attribute query for the row $i$ has been queried.  In other words, the prefix does not count towards the query complexity of the algorithm. 

\begin{remark}
    We remark that one can also define a different query model in which the prefix is released to the adversary whenever the $i$-th row is queried and our results still hold. 
\end{remark}

For completeness, we first show that the Shamir encoding $SS_d$ defined in \cref{def:ss-bit} is $(q(d),d,0)$-secure (as defined in \cref{def:sec-game}) where $q(d)=d$. 
\begin{theorem} \label{thm:sec-A} 
%The scheme $SS_{d'}$ is  $(d',d',0)$-secure where $d'=101d$. 
The scheme $SS_{d}$ is  $(d,d,0)$-secure.  
\end{theorem}
\begin{proof}
Let ${x}\in \{0,1\}^{d}$ and field $\bbF$ s.t. $|\bbF| > {4d}$. \sloppy Recall the secret sharing scheme $SS_{d}({x})= (\alpha_1, \ldots, \alpha_{d}, \{(\alpha_j, p(\alpha_j))\}^{3d}_{j=d+1})$ defined in \cref{def:ss-bit}. We describe two experiments below where the Real World experiment simulates the view of the adversary and the Ideal World experiment just randomly outputs field elements. We will show that these two experiments are perfectly indistinguishable, and the security claim follows. 

\noindent{\bfseries Real World(${x}$).}
\begin{enumerate}
    \item Query $SS_{d}({x})$ for the first $q(d)=d$ pairs of coordinates and let the answers be the prefix $\alpha_1, \ldots \alpha_d$ and $\{(\alpha_{j+d}, z_{j+d})\}_{j \in [d]}$. 
    \item Output $\alpha_1, \ldots \alpha_d$ and $\{(\alpha_{j+d}, z_{j+d})\}_{j \in [d]}$
\end{enumerate}

\noindent{\bfseries Ideal World(${x}$).}    
\begin{enumerate}
    \item Uniformly sample $\alpha'_1,\ldots,\alpha'_{d},\alpha'_{d+1},\ldots,\alpha'_{2d}, r_{d+1}, \ldots, r_{2d}$ from $\bbF$.
    \item Output $\alpha'_1, \ldots \alpha'_d$ and $\{(\alpha'_{j+d}, r_{j+d})\}_{j \in [d]}$
\end{enumerate}

Since by construction, the first $d$ pairs of coordinates returned by $SS_{d}$ and the prefix of size $d$  correspond to $3d$ random field elements, the view of the $\textbf{Real World}$ is therefore just the uniform distribution on $3d$ field elements and thus is identical to that of the view of the $\textbf{Ideal World}$.   
\end{proof}

We present our main lower bound result in \cref{thm:main-ss}. Before we proceed, we first demonstrate the existence of a DP linear-time algorithm and non-DP sublinear-time algorithm for $\cP_{SS}$ below. 
\begin{theorem}\label{thm:dp-pss}
    There exists a $(\eps,\delta)$-DP algorithm that solves the problem $\cP_{SS}$ with $O(d \sqrt{d})$ attribute query complexity and $(1/10,1/10)$-accuracy. 
\end{theorem}
\begin{proof}
     On input $D_S$, the algorithm queries the entire dataset via attribute queries, i.e., it makes $dn = O(d \sqrt{d})$ queries. Given $SS(x_i) = (\alpha^{(i)}_1, \ldots, \alpha^{(i)}_{d}, \{(\alpha^{(i)}_j, p_i(\alpha^{(i)}_j))\}^{3d}_{j=d+1})$ for a row $i \in [n]$, the algorithm first recovers the polynomial $p_i$ of degree $2d-1$ by doing Lagrange Interpolation over the $2d$ points given by $\{(\alpha^{(i)}_j, p_i(\alpha^{(i)}_j))\}^{3d}_{j=d+1}$. Then the original row $x_i$ is obtained by evaluating $(p_i(\alpha^{(i)}_1), \ldots, p_i(\alpha^{(i)}_d))$. Once the original rows $x_i,\ldots,x_n$ are recovered in this manner, the algorithm can release the one-way marginals by adding Gaussian noise as detailed in \cref{lem:gauss}.
\end{proof}

\begin{theorem}\label{thm:subtime-pss}
    There exists a sublinear-time algorithm that solves the problem $\cP_{SS}$ with $O(d \log d)$ attribute query complexity and $(1/10,1/10)$-accuracy. 
\end{theorem}%\todo{check the accuracy here}
\begin{proof}
The algorithm makes $O(d \log d)$ attribute queries and performs the same decoding procedure as outlined in the proof of \cref{thm:dp-pss} to obtain the underlying $\log(d)$ rows and computes the one-way marginals on this subset of rows. The accuracy of this algorithm is a simple corollary of Hoeffding bounds. Recall that $n \in \Omega(\sqrt{d})$, so the problem size is $\Omega(d \sqrt{d})$.
\end{proof}

\begin{theorem} [Main Theorem]\label{thm:main-ss}
Any algorithm that solves the problem $\cP_{SS}$ with $s$ attribute query complexity, $(1/3,1/3)$-accuracy and $(O(1),o(1/n))$-DP must have $s=\Omega(d\sqrt{d/\log(d)})$.
\end{theorem}
%\todo{$\delta=o(\frac{\log(d)}{\sqrt{d}}$, setting $\delta=o(\frac{1}{n}$ would be smaller than the actual range but still correct and perhaps cleaner? }
In order to prove \cref{thm:main-ss}, we follow a similar strategy as presented in the warm-up \cref{sec:ro}. Given an $(n,d,c)$-FPC, we first show how to construct a $c$-row distribution and an adversary $\cB$ that can identify a user in the coalition of size $c$ in \cref{thm:main-pro-reid-ss}.  
\begin{theorem}\label{thm:main-pro-reid-ss}
  For every $n,d \in \bbN$, $\xi \in [0,1]$ and $c \leq n$, if there exists an $(n,d,c)$-fingerprinting code with security $\xi$, then there exists a distribution on $c$-row databases $\cC_\cS$,  a row permutation $\pi_R:[n] \to [n]$ and an adversary $\cB$ for every randomized algorithm $\cA$ with attribute query complexity $cd'$ and $(1/3,1/3)$-accuracy for $\cP_{SS}$ such that
   \begin{enumerate}
       \item \label{it:ss-sec}$\Pr_{C'_S\leftarrow \cC_\cS} [\cB^{\cA}(C'_S ) = \perp] \leq \xi$       
       \item \label{it:ss-sound} For every $i \in [c]$, $\Pr_{C'_S\leftarrow \cC_\cS}[\cB^{\cA}(C'_{S_{-i}}) = \pi_R^{-1}(i)] \leq \xi$.% where $\pi_R:[n] \to [n]$ is a permutation. 
   \end{enumerate}
   where $d'=101d$ and the probability is over the random coins of $\cB$ and the choice of $C'_S$.
   
%   with probability at least $1-\xi$, $\cB$ can successfully re-identify a user in a coalition of size $c$ defined by $D_c \leftarrow \cD_c$.    
   %a distribution on $c$-row databases $D_c \in (\{0,1\}^d)^c$ that is $(\xi,\xi)$-reidentifiable from $(1/3,0, 1/3)$-accurate answers to $\cP_{RO}$. 
\end{theorem}

 Let $(Gen, Trace)$ be the promised $(n,d,c)$-fingerprinting code in the theorem statement. We first construct a PR-PPC $(n,d,c,\ell)$-FPC with $\ell:=100d$ (see \cref{sec:prppc} for details). 
 
 The distribution $\cC_\cS$ on $c$-row databases is implicitly defined through the sampling process below
 \begin{enumerate}
 \item Let $C' \leftarrow Gen'(n,d')$ (see \cref{alg:gen'}) where $C' \in \{0,1\}^{n \times d'}$ and $d'=101d$. 
     \item Let $C'_S =(x_1, \ldots, x_c)\in \{0,1\}^{c \times d'}$ be the first $c$ rows of $C' \in \{0,1\}^{n\times d'}$
     \item Output $C'_S$
 \end{enumerate}
    %from which a $c$-row database is sampled below.  
    % \begin{enumerate}
    %     \item $(C',st') \leftarrow Gen'(\ell)$ with $\ell= d^2$
    %     %\item $C' \leftarrow \textsf{PAD}(C^\pi,\ell)$
    %      \item $C'_{\{1,\ldots,c\}} \leftarrow$ First $c$ rows from $C'$ 
    %     \item Output $C'_{\{1,\ldots,c\}} =\{x_1, \ldots, x_c\}\in \{0,1\}^{c \times d'}$
    % \end{enumerate}
Next we define the privacy adversary $\cB$.

{\bfseries Adversary $\cB$ Algorithm. } Let $\bbF$ be a finite field of order $q'$ where $q' >4d'$. Adversary $\cB$ receives $C'_S=(x_1, \ldots, x_c)$ as input and feeds the algorithm $\cA$ an input instance $C'_\cB$ of $\cP_{SS,d'}(C')$ by simulating answers to attribute queries made by $\cA$ as described in Step \ref{item:b-sim-ss} below. $\cB$ then uses the rounded answer returned by $\cA$ (Step \ref{item:a-rec}) to obtain an individual in the coalition by invoking $Trace'$ in Step \ref{item:trace}. 

\begin{enumerate}
	\item Initialize $q_i = 0$ for each row $i \in [n]$ and initialize a counter $t=0$.
	\item Simulate the oracle algorithm $\cA$ with query access to an $(n \times 2d')$ database $C'_\cB$:
	\begin{enumerate}
		\item \label{item:b-sim-ss} When $\cA$ makes a fresh query $(i,j)$, update $q_i = q_i+1$ and 
		\begin{itemize}
			\item If $q_i \leq d'$, then set $b = q_i + d'$. Respond with a random pair of field elements $(\alpha_{b}^{(i)}, z_{b}^{i})$. Record this tuple.
			\item If $q_i = d'+1$, then
			\begin{enumerate}
				\item Increment $t$ by one.
				\item Define the entire polynomial $p_i$ randomly, subject to the constraints that it is consistent with row $x_t$ and the previous responses sent for row $i$: $p_i(\alpha_j^{(i)}) = x_{t,j}$ for $j \leq d'$ and $p_i(\alpha_{b}^{(i)}) =  z_{b}^{i}$ for $j>d'$. 
				\item Send $\cA$ the response $(\alpha_{j+d'}^{(i)}, p_i(\alpha_{j+d'}^{(i)}))$.
			\end{enumerate}
			\item If $q_i > d'+1$, then the polynomial $p_i$ is already defined. Send the response  $(\alpha_{j+d'}^{(i)}, p_i(\alpha_{j+d'}^{(i)}))$.
		\end{itemize}
		\item \label{item:a-rec} When $\cA$ outputs a vector $\mathbf{a} \in [0,1]^{d'}$, round its entries to ${0,1}$ and call it $\bar{\mathbf{a}}\in \{0,1\}^{d'}$.
		\item \label{item:trace} Return $Trace'(\bar{\mathbf{a}})$.
	\end{enumerate}
\end{enumerate}

We emphasize that although algorithm $\cA$ can make attribute queries to more than $c$ rows, the adversary $\cB$ never defines a secret sharing polynomial for more than $t \leq c$ rows of the input $C'_S$.

\begin{lemma}\label{clm:success-A-ss}
    Suppose $\cA$ is $(1/3,1/3)$-accurate for $\cP_{SS,d'}(C')$. Then the rounded output $\mathbf{\bar{a}}$ from algorithm $\cA$ is a feasible codeword for both $C'$ and $C'_S$ with probability at least $1-\frac{1}{3}-\frac{1}{100}$.  
    
    In other words, with probability at least $1-\frac{1}{3}-\frac{1}{100}$, $\mathbf{\bar{a}} \in F(C')$ and $\mathbf{\bar{a}} \in F(C'_S)$.
\end{lemma}
\begin{proof}
Assuming that $\cA$ is $(1/3,1/3)$-accurate for $\cP_{SS,d'}(C')$, we first show that $\mathbf{\bar{a}}$ is a feasible codeword for $C'$ with probability at least 2/3. By the accuracy guarantee of $\cA$, we know that for any column $i_j$ with probability at least 2/3, $\vert \mathbf{a}_{i_j} - a_{i_j} \vert \leq 1/3$ where $a_{i_j}$ is the actual one-way marginal for column $i_j$. Thus for any column $i_j$ of all 1's in $C'$, $\mathbf{a}_{i_j} \geq 2/3$ which means $\mathbf{\bar{a}}_{i_j}=1$, thus satisfying the marking condition. A similar argument holds for the case when a column is all 0's.

Next, using the fact that we use a PR-PPC $(n,d,c, 100d)$-FPC and that $\mathbf{\bar{a}}\in F(C')$ with probability at least 2/3, we can invoke \cref{lem:win-sample} which states that the feasible sample property is satisfied by our PR-PPC FPC construction. Note that in our case, the sampling algorithm described in \cref{def:feasible-sample} is $\cA$ together with the postprocessing step of rounding the output of $\cA$. Also, even though $\cA$ takes the entire dataset as input, it effectively only has access to the rows of the underlying sample via queries to $\cB$ and thus satisfies the properties required in \cref{def:feasible-sample}. 
In particular, recall that the adversary maintains the invariant $t \leq c$.
\cref{lem:win-sample} states that with probability $\leq \frac{1}{100}$, $\bar{\mathbf{a}}$ is not a feasible codeword for $C'_S$. By a union bound we have that $1-\frac{1}{3}-\frac{1}{100}$, $\mathbf{\bar{a}}$ must be a feasible codeword for $C'_S$. 
\end{proof}

\begin{proof}[Proof of \cref{thm:main-pro-reid-ss}]
From the above \cref{clm:success-A-ss}, we have that $\cA$ is $(1/3,1/3)$-accurate for $\cP_{SS,d'}(C')$ implies that $\bar{\textbf{a}}$ is a feasible codeword for $C'_S$.  By the security of the underlying $(n,d,c)$-fingerprinting code and the corresponding security guarantee of the PR-PPC $(n,d,c,100d)$-FPC given by \cref{cor:rpppc}, we have that $\Pr[\bar{\textbf{a}} \in  F(C'_S) \wedge Trace'(\bar{\textbf{a}} ) = \perp ]\leq \xi$. Since $\cB$ releases the output of $Trace'(\mathbf{\bar{a}})$, the event $\cB^{\cA}(C'_S) = \perp$ is identical to $Trace'(\bar{\textbf{a}} ) = \perp$. Thus Property \ref{it:reid-sec} of the theorem statement which states that the probability that $\cB$ outputs $\perp$ is bounded by $\xi$ follows. Property \ref{it:reid-sound} follows directly from the soundness property of the fingerprinting code. 
\end{proof}

\begin{corollary}\label{corol:atleast-A}
    $\cA$ must make at least $c\cdot d'$ attribute queries to $C_\cB$ to obtain $c$ rows of $C'_S$ where $d'=101d$. 
\end{corollary}
\begin{proof}
    Recall that \cref{thm:sec-A} states that $SS_{d'}$ is $(d',d',0)$-secure where $d'=101d$. Thus, in order to obtain each row of $C'_S$, $\cA$ must make at least $d'$ cell queries. The statement follows from the fact that $\cA$ queries for $c$ rows in total.

\end{proof}

\begin{lemma}\label{lem:acc-dp-ss}
Suppose there exists a distribution on $c \leq n$ row databases according to \cref{thm:main-pro-reid-ss}. Then there is no $(\eps,\delta)$-DP algorithm $\cA$ with row query complexity $c$ that is $(1/3,1/3)$-accurate for $\cP_{SS}$ with $\eps=O(1)$ and $\delta=o(1/c)$.  
\end{lemma}
\begin{proof}
    Suppose $C'_S$ is sampled from the distribution on $c$-row databases $\cC_\cS$ and $\cB$ is the adversary from \cref{thm:main-pro-reid-ss}. From the lemma statement we know that $\cA$ is $(1/3,1/3)$-accurate, thus using \cref{clm:success-A-ss} and \cref{thm:main-pro-reid-ss}, we have that $\Pr[\pi_R(\cB^\cA(C'_S)) \in [c]] \geq 1-\frac{1}{3}-\frac{1}{100}-\xi \geq \Omega(1)$. By an averaging argument, this means that there exists some $i^* \in [c]$ for which $\Pr[\pi_R(\cB^\cA(C'_S))=i^*] \geq \Omega(1/c)$. However, if $\xi = o(1/c)$ by Property \ref{it:ss-sound} in \cref{thm:main-pro-reid-ss} we have that $\Pr[\pi_R(\cB^{\cA}(C'_{S_{-i^*}})) =i^*] \leq \xi = o(1/c)$. 
    
    In other words, the probability of $\cB^\cA$ outputting a fixed output $i^*$ on neighboring input databases $C'_{S}$ and $C'_{S_{-i^*}}$ is different which violates $(\eps,\delta)$-DP for any $\eps=O(1)$ and $\delta=o(1/c)$. We note here that since $\cA$ can make at most $c$ row queries, the DP guarantee for $\cA$ must hold for any neighboring sample of $c$ rows. Since $\cB$ does some postprocessing of the output from $\cA$, and we have shown that $\cB$ cannot be $(\eps,\delta)$-DP, this implies that $\cA$ cannot be $(\eps,\delta)$-DP for any $\eps=O(1)$ and $\delta=o(1/c)$.   

\end{proof}

\begin{proof}[Proof of \cref{thm:main-ss}]
Recall that \cref{lem:acc-dp-ss} states that if there exists a distribution $\cC_\cS$ on $c \leq n$ row databases, then there is no $(\eps,\delta)$-DP algorithm $\cA$ that is $(1/3,1/3)$-accurate for $\cP_{SS}$ with $\eps=O(1)$ and $\delta=o(1/c)$. From \cref{thm:main-pro-reid-ss},  such a distribution can be constructed from an $(n,d,c)$-fingerprinting code. Finally, invoking the Tardos construction for fingerprinting codes in \cref{thm:tardos}, we get that the row query complexity must be $c=\Omega(\sqrt{d/\log(d)})$. Using \cref{corol:atleast-A}, we know that the cell query complexity must be at least $c \cdot d' \geq \Omega(d\sqrt{d/\log(d)})$ where $d'=101d$. 
\end{proof}

\bibliographystyle{plainurl}
\bibliography{references}

\end{document}